\documentclass[a4paper,11pt]{article}

\usepackage{authblk}
\usepackage{a4wide}
\usepackage[utf8]{inputenc}
\usepackage{amsmath}
\usepackage{amsfonts}
\usepackage{amssymb}
\usepackage{mathrsfs}
\usepackage{amsthm}
\usepackage{color}
\usepackage{stmaryrd}
\usepackage[hidelinks]{hyperref}

\usepackage{latexsym,enumerate}
\usepackage{xcolor}
\usepackage{graphicx}
\usepackage{epstopdf}

\usepackage[T1]{fontenc}
\usepackage[utf8]{inputenc}
\usepackage[english]{babel}

\newcommand{\dx}[0]{\,\mathrm{d}}

\newcommand{\fbsde}[0]{}
\newcommand{\esssup}[0]{\mathrm{ess}\,\mathrm{sup}}

\newcommand*{\wh}{\widehat}
\newcommand*{\wc}{\check}

\newcommand*{\IR}{\mathbb{R}}

\newcommand*{\R}{\mathbb{R}}

\def \1{\mathbf{1}}

\newcommand{\be}{\begin{eqnarray*}}
\newcommand{\ee}{\end{eqnarray*}}
\newcommand{\ben}{\begin{eqnarray}}
\newcommand{\een}{\end{eqnarray}}
\newcommand{\bi}{\begin{itemize}}
\newcommand{\ei}{\end{itemize}}

\definecolor{color2}{gray}{0.7}

\newtheorem{thm}{Theorem}[section]

\newtheorem{lemma}[thm]{Lemma}
\newtheorem{propo}[thm]{Proposition}

\theoremstyle{definition}

\newtheorem{definition}[thm]{Definition}

\newtheorem{remark}[thm]{Remark}

\title{Evaluation of equity-based debt obligations}
\author[1]{Alexander Fromm\thanks{A. Fromm acknowledges support from the \emph{German Research Foundation} through the project AN 1024/4-1.}\thanks{alexander.fromm@uni-jena.de}}
\affil[1]{\small Institute for Mathematics, University of Jena, Ernst-Abbe-Platz 2, 07743 Jena, Germany}

\begin{document}

\maketitle

\begin{abstract}
We consider a class of participation rights, i.e.\ obligations issued by a company to investors who are interested in performance-based compensation. Albeit having desirable economic properties equity-based debt obligations (EbDO) pose challenges in accounting and contract pricing. We formulate and solve the associated mathematical problem in a discrete time, as well as a continuous time setting. In the latter case the problem is reduced to a forward-backward stochastic differential equation (FBSDE) and solved using the method of decoupling fields.
\end{abstract}

\vspace{0.5cm}


\section*{Introduction}

Equity-based debt obligations (EbDOs) are a form of participation rights: They constitute a legal arrangement between an investor and a company according to which the investor is promised a share of the company's profits. Contrary to common shares an EbDO does not constitute co-ownership of a company. It is legally a form of debt. At the same time the volume of the debt is not fixed but dependent on the future performance of the company.

EbDOs are unique among other types of participation rights in that they offer the most direct and reliable access to the company's future equity: An EbDO is defined as any obligation according to which the money owed is an increasing function of the company's equity at some moment in the future (maturity). Here the equity of a company is defined as the sum of all its assets minus the total volume of all its outstanding debt at a given moment in time.

We discuss the motivation behind EbDOs in Section \ref{economicimportance} and compare them to other forms of participation. A key advantage of EbDOs is that the company and the respective investor both win or loose at the same time. The possibility of one of the sides taking advantage of the other is, thus, greatly reduced.

However, there is an obstacle to using EbDOs in practice: While the pay-off of an EbDO is a function of the equity, the equity is also a function of the expected pay-offs of outstanding EbDOs, since in the calculation of the equity all debt must be considered. Thus, the equity of a company together with the value of outstanding EbDO debt is not defined explicitly, but only given implicitly via the pay-off functions of the outstanding EbDOs.

From a mathematical point of view we are faced with the following problem: Given the \emph{gross equity} of a company, i.e.\ the value of its assets minus the value of all non-performance-based liabilities, and given a list of pay-off functions for the outstanding EbDOs calculate explicitly how the gross equity is to be divided between the company and the investors, such that the share remaining with the company, which is the \emph{net equity}, yields the shares of the investors when applying the pay-off functions to it. A more fundamental problem is to show that there is a unique solution to this problem in the first place. \\
Being able to resolve this is essential to the applicability of EbDOs in practice: Firstly, companies are usually legally obliged to know and to report their equities. More importantly, a company must know its equity at the time an EbDO matures, as otherwise the volume of the pay-off cannot be determined. Finally, it is highly useful to know how a change in the gross equity impacts the actual equity: This would for instance allow the management of the company to determine to what price a new EbDO can be sold. The price must be such that the increase in the company's assets due to raised funds outweighs the expansion in outstanding EbDO debt such that the actual equity increases as a consequence of the transaction.

We solve the problem of evaluating EbDO debt and calculating the net equity in two different settings: under the discrete time model and the continuous time model. While the two models have their advantages and disadvantages both methods can be used in practice. The calculations in the discrete time case, however, are mathematically less challenging and more straightforward to implement. The continuous time model, however, is more flexible and allows extensions to more complex problems.

In the course of studying the continuous time model we reduce the EbDO debt evaluation problem to a forward-backward stochastic differential equation (FBSDE). Unfortunately, the resulting system is coupled, i.e.\ neither the forward nor the backward equation can be simulated independently of the other. Furthermore, such coupled systems are not necessarily well-posed. It is a longstanding challenge to find conditions guaranteeing that a given fully coupled FBSDE possesses a solution. Sufficient conditions are provided e.g.\ in \cite{Ma1994}, \cite{Pardoux1999}, \cite{ma:yon:99}, \cite{peng:wu:99}, \cite{Delarue2002}, \cite{ma:wu:zhang:15} (see also references therein). The method of decoupling fields, developed in \cite{Fromm2015} (see also the precursor articles \cite{ma:yin:zhan:12} 
and \cite{ma:wu:zhang:15}), is practically useful for determining whether a solution exists.
A decoupling field describes the functional dependence of the backward part $Y$ on the forward component $X$. If the coefficients of a fully coupled FBSDE satisfy a Lipschitz condition, then there exists a maximal non-vanishing interval possessing a solution triplet $(X,Y,Z)$ and a decoupling field with nice regularity properties. The method of decoupling fields consists in analyzing the dynamics of the decoupling field's gradient in order to determine whether the FBSDE has a solution on the whole time interval $[0, T]$.
The method can be successfully applied to various problems involving coupled FBSDE: In \cite{Proemel2015} solutions to a quadratic strongly coupled FBSDE with a two-dimensional forward equation are constructed to obtain solutions to the Skorokhod embedding problem for Gaussian processes with non-linear drift. In \cite{FrommImkeller2017} the problem of utility maximization in incomplete markets is treated for a general class of utility functions via construction of solutions to the associated coupled FBSDE. In the more recent work \cite{doi:10.1137/17M1152401}, the method is used to obtain solutions for the problem of optimal control of diffusion coefficients. In this paper we follow a similar methodology in showing that our FBSDE is in fact well-posed.

This paper is structured as follows: In Section \ref{economicimportance} we discuss EbDOs from a purely economic point of view and consider their key features. In Section \ref{probfordisc} we mathematically formulate the EbDO evaluation problem in discrete time. In Section \ref{solving discr} we solve the discrete time problem in the sense that its well-posedness is shown and a simple numerical scheme to calculate solutions is deduced. In Section \ref{probforcont} the continuous time problem is formulated in the form of a coupled FBSDE. Since we rely on the method of decoupling fields to study this system we briefly introduce this method and its underlying theory in Section \ref{sec:decfields}. Finally, in Section \ref{solving conti} existence and uniqueness of solutions to the FBSDE from Section \ref{probforcont} is shown. Moreover, the solution is obtained explicitly for a simple illustrating example. Note that we do not provide a numerical scheme for approximating solutions to the FBSDE introduced in Section \ref{probforcont}: numerical treatment of FBSDE is a separate topic and is usually considered in a more general context.

\section{EbDOs in comparison with other forms of participation}\label{economicimportance}

As touched upon in the introduction the motivation behind studying equity-based debt obligations is what one may refer to as the \emph{investor's participation problem}: An investor owns assets that could help a company to meet its production goals. Assume that the investor does not need or use these assets himself or herself at a given moment and could provide them to the company assuming there is sufficient return on his or her investment. Clearly, this return should come from claims to the company's assets based on a legal arrangement between the investor and the company. The problem consists of designing the contract such that both sides benefit or may expect to benefit on average.

The \emph{first precondition} for any successful arrangement is of course that the company has a sound business strategy such that assets inside the company may be expected to grow due to profits. Note that such profits are often rooted not only in the company's existing know-how and expertise, but also in the fact that through investors' participation the company is able to amass various resources needed to implement a non-trivial production scheme in the first place. In other words, the pooling of resources in itself, e.g.\ through the issuance of participation rights, may significantly contribute to the business's capability and profitability.

Mathematically, profits may be defined as changes of the company's equity in time. One might further distinguish between gross and net profits depending on what is meant by equity (e.g.\ gross equity or net equity). Apart from assumptions on profitability the \emph{second precondition} for a successful arrangement is the existence of a clear definition of what it means that the investor receives a share of the company's future profits. It is natural to agree to a payoff which is a function of the overall profit over the period of time beginning with the moment the arrangement is made and ending at a well-defined maturity. Since the current equity is known and only the future equity is a random variable, this is equivalent to assuming that the payoff is a function of the future equity. Unlike other participation rights an EbDO is a function of the \emph{net equity} rather than of any other notion of equity. Note here, that the net equity is the true equity of the company as this is what remains after \emph{all} existing obligations are considered and subtracted, including outstanding EbDOs themselves. We further assume that this payoff function is increasing: the larger the future net equity the more the investor will get at maturity. It is also natural to assume that there is no payoff in case of zero net equity. A typical EbDO payoff function is
$$ h(y):=\alpha \cdot (y-y_0)^+, $$
where $y_0\geq 0$, $\alpha>0$ are fixed constants and where $\cdot^+$ refers to the positive part of a real number. In other words, the payoff is equal to $\alpha (y-y_0)$ if the future equity $y$ is larger than $y_0$, the latter being typically set to the equity at the time the contract is entered. If $y\leq y_0$, on the other hand, there is no payoff at all. Now assume that for an EbDO with payoff function $h$ maturity is reached. Then the gross equity $x$, i.e.\ the sum of all assets minus all fixed, i.e.\ non-performance-based, liabilities is known or can be straightforwardly calculated. The value $y$ is still unknown and is implicitly given by the condition $y+h(y)=x$, assuming there is only one EbDO and no other performance based debt. If $y_0=0$ and $x>0$ then the unique solution is $y=x\frac{1}{1+\alpha}$. Accordingly, $h(y)=x\frac{\alpha}{1+\alpha}$. Note that $y$ is positive regardless of how large $\alpha>0$ was chosen. Thus, there is no limit to the volume of EbDOs a company can issue.

Apart from EbDOs there are other participation schemes used for essentially the same economic purpose. As examples let us name common shares, preferred stock and participation rights where the payoff is a function of the gross equity or changes of the gross equity over periods of time. As a specific example one might consider a payoff $h(x):=\alpha (x-x_0)^+$, where $x\geq 0$ refers to the gross equity of the company at a fixed future moment in time. In this case the net equity would be $y=x-h(x)$, assuming there is no other performance based debt. Observe that $\alpha>0$ should not exceed $1$ as otherwise bankruptcy occurs for sufficiently large $x$!

There are various reasons for EbDOs being a superior solution compared to the aforementioned alternatives, in the sense that EbDOs as a class of arrangements provide a better deal for both the investor and the company. We now briefly discuss the key advantages: \vspace{2mm} \\
{\bf Predictability and measurability:} As we shall see in the following sections the investor can calculate, based on his or her expectations about the company's (future) profitability, explicitly how much he or she will receive as EbDO holder on average. Other moments of this random variable can be calculated as well. Apart from that, the investor knows when the payoff occurs as the maturity is agreed on in the contract. This is a significant advantage over e.g.\ common shares: In the latter case the investor might not know when dividend payments will occur and whether there will be any dividend payments at all. Also, such payments are not necessarily tied to the company's performance: The company making profits does not immediately imply that they will be paid out in the form of dividends. This is likely to depend on decisions made by the management and/or other investors (shareholders) and is often merely insignificantly influenced by a given investor. This might be of major concern for minority shareholders in particular. 
In case the company is not paying any dividends and there are no reliable mechanisms to force it to do so, holding shares of such a company might be of interest for pure speculators only and not for actual investors.\vspace{2mm} \\
{\bf Incentivisation:} A key advantage of EbDOs is that they prevent conflicts of interest between the company and the respective investors. This is because the payoff is an increasing function of what remains in the company after the payoff (and all other payoffs and/or expected values of future payoffs). As a consequence there is no action a company can take to reduce the payoff without reducing its own net equity. Conversely, there is no action the investor can take to increase its payoff without also helping the company to grow. In other words, the company and the investor "sit in the same boat". By extension all EbDO holders "sit in the same boat".  The net equity of the company serves as the reference value everyone seeks to maximize.\\
The above property of EbDOs cannot be underestimated as it makes them highly suitable for compensations of executives and other persons exercising control over the company or otherwise influencing its economic performance. It pushes these individuals towards cooperation as everyone is interested in increasing the same value, which is the net equity.
\vspace{2mm} \\
{\bf Flexibility:} Unlike common shares an EbDO can be tailor made to suit the needs of individual investors by choosing the respective payoff function and maturity date appropriately to reflect the preferences of a particular investor.  \vspace{2mm} \\
{\bf Easy issuance:} A major problem with common shares is dilution: The more shares are issued the more the value of already issued shares is expected to decrease, especially if shares are sold at a relatively low price. This means that, on the one hand, the equity of the company increases due to the funds raised through public offerings, but on the other that existing investors do not necessarily benefit from this increase. There is no straightforward criterion to decide under which circumstances it is or is not appropriate for the management to issue new shares in light of its commitment to protect the interests of existing investors. \\
With EbDO's, however, this problem does not arise: Any action that increases the company's net equity is good both for the company and existing EbDO holders. In other words, the company's management is free to just concentrate on raising the net equity, e.g.\ through the issuance of EbDOs. It must merely make sure that no EbDO is sold below the so-called \emph{risk-neutral price} (see Section \ref{solving discr}).

Another problem along these lines occurs with participation rights, where the gross equity is used as the underlying: There is a natural limit on how many such participation rights can be issued as all of the gross equity is eventually claimed by and divided between existing investors and the remaining net equity becomes insignificant in comparison. At the same time the company might become unable to raise more funds through the issuance of additional participation rights as existing investors would still get their share of the gross equity and its increments, even if the contribution of newcomers is more significant. This may lead to an "investment deadlock", as the company is unable to offer a reasonable deal to new investors, i.e.\ a deal which would not amount to immediate exploitation of new investors by the existing. \\
This problem, however, does not arise with EbDO's as there is no natural limit to the overall volume of outstanding EbDOs. Whatever the existing EbDO - debt structure is it is impossible for the gross equity to be completely "consumed" such that the net equity reaches zero: Otherwise the EbDO - debt would be zero as well and, thus, the net equity actually equal to the gross equity. In other words, there is always room to accommodate new investors. Also, unlike in the case of participation rights based on the gross equity, the company cannot be driven into bankruptcy through EbDOs alone. On the contrary, EbDOs have an amortizing effect and help to protect the company from insolvency. \vspace{2mm} \\
{\bf Neutrality in terms of corporate governance:} Unlike common shares EbDOs do not go along with any controlling interest. In other words, they do not constitute co-ownership of the company nor do they entitle to any right to influence the management's decision making or the composition of the management. Although at first one might see it as a disadvantage, there is actually none, since the respective investor might still be allowed to exercise control over the company due to additional arrangements or because of an existing or scheduled role as a manager. \\ In general, it is neither necessary nor always desirable for an investor to exercise control and interfere in the company's operations. Whatever the internal corporate governance procedures are, they should be designed to ensure economically sound decision making. In particular, the company should be run by whoever is most competent in achieving stated business objectives. Obviously, this is not necessarily the person or the group of people who provide the largest amount of assets. It is worth noting that especially for significantly diversified investors it is neither desirable nor possible to bear responsibility for the management of every company in the portfolio. On the contrary, an investor might be reluctant to invest or even investigate whether the company should be invested in, if the company raises funds through shares or participation rights which lead to co-ownership: the company might end up under the control of future investors who's identities are still unknown such that this investor is being offered to buy a "cat in the sack". \\
To sum up, a company's internal constitution, composition and compliance standards are an important, but completely separate topic. Whatever the chosen corporate design is, it is desirable that investment vehicles used by the company to raise funds do not interfere with internal mechanisms of decision making or at least that this interference is not "hard-coded". This gives freedom to choose the most effective set of internal rules and practices. \vspace{2mm} \\
{\bf Invariance w.r.t.\ jurisdiction:} Unlike common shares the economic meaning of an EbDO is always the same regardless of the jurisdiction. The rights of a holder of common shares and the obligations of the company towards a shareholder, but also the relationship between minority and majority shareholders might vary greatly from jurisdiction to jurisdiction. In addition, there is the threat of new laws being enacted changing the nature of existing investments. A participation right in the form of a well-defined financial obligation with a fixed maturity and pay-off function greatly reduces such uncertainties. \vspace{2mm} \\
{\bf Tax efficiency:} Note that a company usually pays corporate taxes on positive changes of its equity but not necessarily on increments of the gross equity. If an increase of the gross equity occurs an increase of the value of outstanding EbDO-debt occurs at the same time. This reduces the company's net profits and impacts the corporate taxes it has to pay. Depending on the jurisdiction this might make EbDOs (and other participation rights) significantly more attractive in terms of corporate taxation compared to common shares. \vspace{2mm}

To sum up, EbDOs are, from a purely economic view, a highly desirable arrangement as they couple the profit of the investor to the overall economic success, over a given period of time, of the respective company in the most effective and reliable way possible. At the same time the nature of this arrangement creates a non-trivial evaluation and accounting problem: Since the payoff is a function of the net equity and the net equity depends on the payoff neither can be determined independently of the other. In other words, it is not a priori clear how a given gross equity is to be divided between the investors and the company such that an investor's payoff is a prescribed function of what remains in the company. Although for special cases, e.g.\ if there is only one investor and the EbDO matures immediately, the problem has a simple solution, some level of complexity is needed to treat the general case. In the following sections we formulate, study and solve this problem in a rigorous mathematical setting.

\section{Problem formulation in discrete time}\label{probfordisc}

For $n\in\mathbb{N}$ let $0\leq T_1< T_2<\ldots< T_n$ be deterministic future times. In addition, assume that for each $i\in\{1,\ldots,n\}$ there is a pay-off function $h_i:[0,\infty)\rightarrow[0,\infty)$, which determines the total pay-off of all EbDOs which mature at time $T_i$. This pay-off is a function of the net equity $Y_i$ of the company at time $T_i$. Each $h_i$ is monotonically increasing. Furthermore, we assume that $h_i(0)=0$ for all $i$. Otherwise, we could split up the pay-off into a constant part and a monotonically increasing part which starts at $0$. The constant part is then added to the fixed debt (i.e. debt which does not depend on the performance of the company) and is considered when calculating the gross equity $X_0\geq 0$ at time $0$. This value is deterministic and can be calculated by summing up the value of all assets of the company and subtracting all the non EbDO-debt from it. 

Our primary aim is to calculate the equity $Y_0$ of the company at time $0$. In case $T_1=0$ the company would pay out the amount $h_1(Y_0)$ to the holders of the EbDOs which mature at time $T_1$. More generally, the current equity of the company must be known at all times, not just for reporting purposes, but also in order to be able to tell how much must be payed out to the EbDO holders at a given moment. Our secondary aim is to calculate the expected pay-off $\mathbb{E}[h_i(Y_i)]\geq 0$ of the EbDOs which mature at a future time $T_i$ for arbitrary $i=1,\ldots,n$.

We denote by $X_i$ the gross equity at time $T_i$ and by $X'_i:=X_i-h_i(Y_i)$ the gross equity immediately after the payoffs occurring at time $T_i$, $i=1,\ldots,n$. We must have $X'_n=Y_n$ as the gross equity coincides with the net equity after all EbDO debt was served. 
Like the gross equity $X_i$ and the net equity $Y_i$ the random variable $X'_i$ is required to be non-negative.

The basis of our calculation is the current equity before EbDO-debt $X_0=X'_0\in[0,\infty)$, which unlike the values $X_i,Y_i,X'_i$, $i=1,\ldots,n$, is a priori known, and which we refer to as the gross equity at time $0=:T_0$. In addition, we must postulate the dynamics according to which the gross equity evolves between payoffs. The most simple model is the model of a geometric Brownian motion: We assume that $X_{i}$ is equal to $X'_{i-1}\cdot Z_i$, where 
$$ Z_i \sim\mathcal{LN}\left(\left(\mu-\frac{1}{2}\sigma^2\right)\cdot \left(T_{i}-T_{i-1}\right),\,\sigma^2\cdot \left(T_{i}-T_{i-1}\right)\right) $$
has a log-normal distribution with the parameters specified above. $\mu\in\mathbb{R}$ determines the trend of the gross equity and $\sigma\in[0,\infty)$ its volatility in time. Thus, the gross equity does not have a drift (neither to the upside nor to the downside) if and only if $\mu=0$. $Z_i$ is deterministic if and only if $\sigma \left(T_{i}-T_{i-1}\right)=0$, otherwise it is stochastic. $\sigma$ reflects the uncertainty about the future evolution of the gross equity due to the company's intrinsic performance. We assume that the $Z_i$, $i=1,\ldots,n$, are independent random variables. 

For accounting purposes it is necessary to assume that $\mu=0$ since expectations about future profits cannot be included in the calculation of the \emph{current} equity. An investor could set $\mu$ to a positive value to calculate the expected pay-off of his or her EbDO under the assumption that the company grows with the rate $\mu$ per unit of time on average. However, this would be a subjective view, different from a neutral stance to be taken for accounting purposes. In the same context, we must require $Y_i$, $i\in\{0,\ldots,n\}$, to be a \emph{martingale}. In other words the current net equity must be equal to the expectation of a future net equity given the information already available. The martingale property must hold w.r.t.\ the filtration given by $\mathcal{F}_i:=\sigma\left(Z_j,\, j=1,\ldots,i\right)$, $i\in\{0,\ldots,n\}$, where $\mathcal{F}_0$ is trivial.

Now, given $\sigma$, $X_0$ and the pay-off functions $h_i$ the problem of calculating $Y_0$ as well as $\mathbb{E}[h_i(Y_i)]$ under the assumption that $\mu=0$ can be solved which is done in the next section. The solution entails that suitable adapted processes $X_i,Y_i,X'_i$, $i=0,\ldots,n$, such that all of the above is satisfied, exist and are unique in the first place. We refer to such $X_i,Y_i,X'_i$, $i=0,\ldots,n$, as the solution to the \emph{EbDO evaluation problem}. To sum up, such a triplet must satisfy:
\begin{enumerate}
\item $(X_i)$, $(X'_i)$, $(Y_i)$ are non-negative and adapted w.r.t.\ $(\mathcal{F}_i)_{i\in\{0,\ldots,n\}}$,
\item $X_{i} =  X'_{i-1}\cdot Z_i$ a.s.\ for all $i\in\{1,\ldots,n\}$,
\item $X'_{i} = X_i-h_i(Y_i)$ a.s.\ for all $i\in\{1,\ldots,n\}$,
\item $(Y_i)_{i\in\{0,\ldots,n\}}$ is a martingale w.r.t.\ $(\mathcal{F}_i)_{i\in\{0,\ldots,n\}}$,
\item $Y_n=X'_n$ a.s.\ and $X'_0=X_0$, where $X_0\in[0,\infty)$ is given.
\end{enumerate}

\section{Solving the evaluation problem in discrete time}\label{solving discr}

We denote by $\mathrm{Id}$ the identity mapping on $[0,\infty)$. Let us define a function $f_{n-1}:[0,\infty)\rightarrow[0,\infty)$ via
$$f_{n-1}(x'):=\mathbb{E}[(\mathrm{Id}+h_{n})^{-1}(x' Z_n)],\qquad x'\geq 0.$$
Note that the function $\mathrm{Id}+h_{n}$ is equal zero at zero, is strictly increasing and, therefore, is invertible. Its inverse is also strictly increasing and is zero at zero. As a consequence $f_{n-1}:[0,\infty)\rightarrow[0,\infty)$ is strictly increasing such that $f_{n-1}(0)=0$.
Next, we define 
$$f_{n-2}(x'):=\mathbb{E}[(f^{-1}_{n-1}+h_{n-1})^{-1}(x' Z_{n-1})],$$
for arbitrary $x'\geq 0$. Again, $(f^{-1}_{n-1}+h_{n-1})^{-1}$ is well-defined, strictly increasing and vanishes at zero. Therefore, $f_{n-2}$ has the same properties. Similarly, we define recursively
$$f_{i-1}(x'):=\mathbb{E}[(f^{-1}_{i}+h_{i})^{-1}(x' Z_{i})],\qquad x'\geq 0, $$
for every $i=1,\ldots,n-2$. This holds true for $i\in\{n-1,n\}$ as well, after setting $f_n:=\mathrm{Id}$.
At the end of this backwards recursion we obtain $f_{0}:[0,\infty)\rightarrow[0,\infty)$. All $f_{i}:[0,\infty)\rightarrow[0,\infty)$ are strictly increasing and vanish at $0$.

We now claim that $Y_0$ can be calculated simply as $f_0(X_0)$:

\begin{thm}\label{mainresultdiscrete}
There exists a unique solution $X_i,Y_i,X'_i$, $i=0,\ldots,n$ to the EbDO evaluation problem. Furthermore, these processes can be obtained explicitly using the following forward recursion: For $X_0\in[0,\infty)$, set $Y_0=f_0(X_0)$ and $X'_0=X_0$. Then, for $i\in\{1,\ldots,n\}$, set
\begin{align}\label{forwardrecursion}
 X_{i} &:=  X'_{i-1}\cdot Z_i, \nonumber\\
 Y_i &:=  (f^{-1}_{i}+h_{i})^{-1}(X_i), \\
 X'_{i} &:= X_i-h_i(Y_i). \nonumber
\end{align}
\end{thm}
\begin{proof} Firstly, note that a random variable $X'_{i}$ is non-negative if the last two equations of \eqref{forwardrecursion} are satisfied with $X_{i},Y_{i}\geq 0$:  We have $(f^{-1}_{i}+h_{i})(Y_i)=X_i$, so 
$$X'_{i}=X_i-h_i(Y_i)=f^{-1}_{i}(Y_i)\geq 0, \qquad i=1,\ldots,n.$$
This shows in particular that recursion \eqref{forwardrecursion} is well defined and the resulting processes are non-negative. Now let us verify that such processes are in fact a solution to the problem: Clearly, $X_i,Y_i,X'_i$ are adapted (inductive argument). 
Note that $X'_{n}=f^{-1}_{n}(Y_n)=Y_n$. Since the properties $ X_{i} = X'_{i-1}\cdot Z_i$ and $ X'_{i} =X_i-h_i(Y_i)$ are given we must merely show that $Y_i$, $i\in\{0,\ldots,n\}$, is a martingale, i.e.\ $Y_{i-1}=\mathbb{E}[Y_i|\mathcal{F}_{i-1}]$: Using the definition of $f_{i-1}$ we have in fact
$$ \mathbb{E}[Y_i|\mathcal{F}_{i-1}] = \mathbb{E}[(f^{-1}_{i}+h_{i})^{-1}(X'_{i-1}\cdot Z_i)|\mathcal{F}_{i-1}] = f_{i-1}(X'_{i-1}), $$
since $\sigma(Z_i)$ and $\mathcal{F}_{i-1}$ are independent and $X'_{i-1}$ is measurable w.r.t.\ $\mathcal{F}_{i-1}$. Now if $i=1$ we have $f_{i-1}(X'_{i-1})=f_0(X'_0)=Y_0$. 
Otherwise, $X'_{i-1}=f^{-1}_{i-1}(Y_{i-1})$ yields $f_{i-1}(X'_{i-1})=Y_{i-1}$, which verifies the martingale property.

On the other hand, we can show that for any solution to the EbDO evaluation problem, i.e.\ for any three non-negative and adapted processes $X_i,Y_i,X'_i$, $i=0,\ldots,n$, such that the properties  $X_{i} =  X'_{i-1}\cdot Z_i$, $X'_{i} = X_i-h_i(Y_i)$, $Y_n=X'_n$, $X_0=X'_0$ and the martingale property for $(Y_i)$ are satisfied, recursion \eqref{forwardrecursion} must already hold:

Without even using the martingale property we first obtain $$Y_n =  (f^{-1}_{n}+h_{n})^{-1}(X_n)=(\mathrm{Id}+h_{n})^{-1}(X_n)$$ from $Y_n=X'_{n} = X_n-h_n(Y_n)$. Next consider the property
$\mathbb{E}[Y_n|\mathcal{F}_{n-1}]=Y_{n-1}$: Using $Y_n=(\mathrm{Id}+h_{n})^{-1}(X_n)$ and  $X_{n} =  X'_{n-1}\cdot Z_n$ we obtain
$$ Y_{n-1} = \mathbb{E}[(\mathrm{Id}+h_{n})^{-1}(X'_{n-1} Z_n)|\mathcal{F}_{n-1}]=f_{n-1}(X'_{n-1}). $$
This implies $ f^{-1}_{n-1}(Y_{n-1})=X'_{n-1}=X_{n-1}-h_{n-1}(Y_{n-1})$, which then yields \eqref{forwardrecursion} for $i=n-1$ using a straightforward transformation. Similarly, consider the martingale property $\mathbb{E}[Y_{j+1}|\mathcal{F}_{j}]=Y_{j}$ for a $j\in\{1,\ldots,n-2\}$ while assuming that \eqref{forwardrecursion} is already verified for all $i\in\{j+1,\ldots,n\}$. Using $X_{j+1} =  X'_{j}\cdot Z_{j+1}$ we have
$$ Y_{j} = \mathbb{E}[(f^{-1}_{j+1}+h_{j+1})^{-1}(X'_{j} Z_{j+1})|\mathcal{F}_{j}]=f_{j}(X'_{j}). $$
This implies $ f^{-1}_{j}(Y_{j})=X'_{j}=X_{j}-h_{j}(Y_{j})$, which then yields \eqref{forwardrecursion} for $i=j$ using a straightforward transformation. This completes an inductive argument showing that \eqref{forwardrecursion} holds for any $i\in\{1,\ldots,n\}$. Finally, considering $\mathbb{E}[Y_{1}|\mathcal{F}_{0}]=Y_{0}$ we obtain
$$ Y_{0} = \mathbb{E}[(f^{-1}_{1}+h_{1})^{-1}(X'_{0} Z_{1})|\mathcal{F}_{0}]=f_{0}(X'_{0})=f_{0}(X_{0}). $$
\end{proof}

\begin{remark}
It is worth noting that once the functions $f_i$ have been calculated the recursion \eqref{forwardrecursion} can be used to simulate the gross and net equities forward into the future. Using Monte Carlo simulation it is straightforward to obtain estimates for $\mathbb{E}[h_i(Y_i)]$. This expected payoff represents the total estimated value of EbDOs which mature at time $T_i$. By performing the same simulation, but with a trend $\mu$ different from $0$ an investor can calculate or estimate the expected payoff $\mathbb{E}[h_i(Y_i)]$ under the assumption that the company performs with some rate $\mu>0$ due to profits expected on average. This allows to assign to each EbDO two prices: The \emph{risk neutral price}, i.e.\ the expected payoff under the assumption that the assets under the control of the company neither grow nor shrink, and the \emph{market price}, which is the price an investor is willing to pay based on his or her expectations regarding the company's future performance. The risk neutral price is the one to be reported in the company's official balance sheet, while the market price is a subjective value used, for instance, by a profit oriented analyst. These two values can considerably differ, as an EbDO might be essentially worthless if no growth occurs, but might become very valuable if the company grows at a steady pace over a prolonged period of time until the EbDO matures.

This difference between the risk neutral price and the subjective market price is the immediate reason why a deal between the company and the investor takes place and the EbDO is sold for a given price: This sum is, from the point of view of the company, more valuable then the EbDO, since the company must assume a risk neutral view in its books. The investor on the other hand may consider the EbDO more valuable than the fixed amount of money paid to purchase it, because the investor is working with a positive $\mu$ and thereby with a higher price.

Conversely, it is possible that investors predominantly work with a negative $\mu$, due to poor expectations regarding the company's future performance. In this case both the company and the investor might be interested in the opposite transaction, i.e.\ in a buy-back of outstanding EbDOs. This might allow the company to improve its balance sheet by removing outstanding obligations from it, while the investor protects himself or herself from future losses by means of selling currently held positions.
\end{remark}

\begin{remark}
Theorem \ref{mainresultdiscrete} confirms that the key to solving the evaluation problem is obtaining the functions $f_{i}:[0,\infty)\rightarrow[0,\infty)$ and $(f^{-1}_{i+1}+h_{i+1})^{-1}:[0,\infty)\rightarrow[0,\infty)$, $i=0,\ldots,n-1$. Note that the former is obtained from the latter by calculating an expectation. More precisely, $f_{i}(x')=\mathbb{E}[(f^{-1}_{i+1}+h_{i+1})^{-1}(x' Z_{i+1})]$ for any $x'\geq 0$, where $Z_{i+1}$ has a log-normal distribution. This motivates a simple numerical scheme to calculate or approximate the functions $f_{i}$:

Assume that we have two piecewise linear approximations of $f_{i+1}$ and $h_{i+1}$, such that these approximations are increasing and equal zero at zero. We also assume that the piecewise linear approximation of $f_{i+1}$ is strictly increasing. Then the corresponding approximation of $(f^{-1}_{i+1}+h_{i+1})^{-1}$ is also piecewise linear, as inverse functions of piecewise linear functions are themselves piecewise linear. It is also strictly increasing and equal zero at zero. By a slight abuse of notation, let us now imagine that $(f^{-1}_{i+1}+h_{i+1})^{-1}$ is piecewise linear. This means that it is a linear combination of functions which are either equal $1$ or equal the identity on some interval and $0$ everywhere else. So, for a given $x'>0$ the value $\mathbb{E}[(f^{-1}_{i+1}+h_{i+1})^{-1}(x' Z_{i+1})]$ can be calculated as a linear combination of integrals of the form
$$ \int_{a}^b \rho(v)\dx v \qquad\textrm{and/or}\qquad \int_{a}^b e^{v} \rho(v)\dx v, $$
where $-\infty<a<b<\infty$ are constants that are calculated explicitly and where $\rho$ is the density of a normal distribution. This statement is true because $(f^{-1}_{i+1}+h_{i+1})^{-1}$ is piecewise linear and $Z_{i+1}$ is equal to the exponential function applied to a normally distributed random variable. Observe further that expressions of the form $\int_{a}^b \rho(v)\dx v$ and $\int_{a}^b e^{v} \rho(v)\dx v$ can be calculated explicitly using the cumulative function of the standard normal distribution.

Since according to the above the function $f_{i}(x')$ can be calculated explicitly for any given $x'$ it is straightforward to obtain a strictly increasing piecewise linear approximation of $f_{i}$ and then repeat the whole process to obtain an approximation of $f_{i-1}$ and so on until a piecewise linear approximation of $f_0$ is obtained.
\end{remark}

While the above numerical scheme is straightforward to implement, the calculations might be somewhat slow if there is a large number of maturities $T_i$: A new function $f_{i}:[0,\infty)\rightarrow[0,\infty)$ needs to be calculated for each $i$. At the same time, if there is such a high density of maturities the payoff process $j\mapsto \sum_{i=1}^j h_i(Y_i)$ might increasingly resemble a continuous process with a \emph{payoff rate} which is a function of the respective net equity. This motivates what we refer to as the \emph{continuous time model} introduced and studied in the following sections.

\section{Problem formulation in continuous time}\label{probforcont}

Under the continuous time model we assume that the EbDOs do not mature at finitely many points $0\leq T_1< T_2<\ldots< T_n$ in time, but instead postulate that the pay-off due to EbDOs takes place continuously in time according to a rate function $h:[0,T]\times[0,\infty)\rightarrow[0,\infty)$ over a given time period $[0,T]$. In other words the \emph{total} pay-off over the time $[0,t]$, where $t\in[0,T]$, is given by
$$ \int_0^t h(s,Y_s)\dx s, $$
where $Y_s$ is the (a priori unknown) net equity of the company at time $s$. Here $T>0$ is larger than the maturity of every EbDO, such that all outstanding EbDOs mature during the time interval $[0,T]$.
Note that the value $Y_s$ is to be determined or modeled stochastically. As a matter of fact, our primary goal is to calculate $Y_0$. We assume that $h(t,\cdot)$ is monotonically increasing and satisfies $h(t,0)=0$, where $t\in[0,T]$ is arbitrary. 

\begin{remark}
Observe that if the company has an EbDO with a fixed maturity $T_i$ and payoff function $h_i$ on its books then the associated payoff $h_i(Y_{T_i})$ can be approximated by the value 
$\frac{1}{\varepsilon}\int_{T_i-\varepsilon}^{T_i} h_i(Y_{s})\dx s$ with some small $\varepsilon>0$. This approximation fits mathematically the continuous time setting with the payoff rate function being $(s,y)\mapsto\frac{1}{\varepsilon}h_i(y)\mathbf{1}_{[T_i-\varepsilon,T_i]}(s)$.
\end{remark}

In addition to the payoff rate $h$ we have the gross equity $X_0\geq 0$ at time $0$ as the starting point of our calculation. Finally, we need to postulate the dynamics of the gross equity $X_s$ in time. Clearly, the gross equity is continuously diminished due to the pay-off rate $h(s,Y_s)$. Apart from that we assume that there are random fluctuations characteristic for a geometric Brownian motion without a drift. Thus, we obtain that $X$ has the dynamics
$$  X_s=X_0-\int_{0}^s h(r,Y_r)\dx r+\int_{0}^s X_r\cdot\sigma\dx W_r, \qquad s\in[0,T],$$
where $W$ is a Brownian motion and where $\sigma\in[0,\infty)$ is a fixed parameter determining the uncertainty about the future evolution of the gross equity due to the company's intrinsic performance. Note that we work under the assumption that on average the company neither shrinks nor grows in time as expectations about future profits cannot be included in the calculations and known future expenditures were already incorporated in the calculation of $X_0$. This means that the \emph{total wealth} 
$$t\longmapsto X_t+\int_0^t h(s,Y_s)\dx s$$ is a martingale. For accounting purposes we must also postulate that the net equity $Y_t$, $t\in[0,T]$, is a martingale as well. Note that $X_T=Y_T$ holds since we assume that there are no more outstanding EbDOs beyond time $T$. Now, the martingale representation theorem yields
$$ Y_s=X_T-\int_{s}^{T}Z_r\dx W_r, \quad s\in[0,T], $$
with some square-integrable process $Z$. To sum up, we have to solve the following coupled forward-backward system:
\ben\label{fbsde}
\begin{array}{rcl}
X_s&=&X_0-\int_{0}^s h(r,Y_r)\dx r+\int_{0}^s \sigma X_r\dx W_r, \\
 Y_s&=&X_T-\int_{s}^{T}Z_r\dx W_r,\qquad \textrm{a.s. for all }s\in[0,T].
\end{array}
\een
In order to show existence and uniqueness of solutions $X,Y$ we use the so-called the method of decoupling fields, which was designed for the purpose of analyzing coupled systems. We briefly introduce the theory of decoupling fields in the following section. A cornerstone of this method is the construction of a time-dependent random field $u$ which connects $X$ and $Y$ via $u(s,X_s)=Y_s$.
\begin{remark}
Before introducing the theoretical backbone of our analysis and then actually solving the above problem in Section \ref{solving conti} let us point out that strongly coupled FBSDEs in general are a powerful and flexible tool allowing to formulate and study more general and more complicated problems than the one given by \eqref{fbsde}. For instance, one might be interested in a problem where the pay-off rate $h$ also depends on $\omega\in\Omega$ or on the gross equity $X$. It is also possible to study multi-dimensional problems where net equities of different companies are to be determined simultaneously, for instance due to two or more companies holding EbDOs of each other. This may occur if there is a set of affiliated companies forming a group. We reserve such considerations and generalizations to future research and concentrate on the basic time-continuous problem provided by \eqref{fbsde}.
\end{remark}

\section{The method of decoupling fields}\label{sec:decfields}

As a key result of this paper we prove in Section \ref{solving conti} the solvability of (\ref{fbsde}). 
Even under Lipschitz assumptions for $h$, it is not trivial to show well-posedness of (\ref{fbsde}) due to
its coupled nature. By this we mean that the forward equation, which describes the dynamics of $X$, depends on $Y$ via $h$, while the backward equation, describing the dynamics of $Y$, depends on $X$ via the condition $Y_T=X_T$. This means that neither of the two processes can be simulated or calculated independently of the other. Furthermore, coupled systems are not always solvable, even under Lipschitz conditions. It is, thus, necessary to take more subtle structural properties into account to conduct the proof. 
Our argumentation will be based on the method of decoupling fields which we briefly sum up in this section.
\vspace{1mm}

For a fixed finite time horizon $T>0$, we consider a complete filtered probability space $(\Omega,\mathcal{F},(\mathcal{F}_t)_{t\in[0,T]},\mathbb{P})$, where 
$\mathcal{F}_0$ consists of all null sets, $(W_t)_{t\in[0,T]}$ is a $1$-dimensional Brownian motion and $\mathcal{F}_t:=\sigma(\mathcal{F}_0,(W_s)_{s\in [0,t]})$ with $\mathcal{F}:=\mathcal{F}_T$. The dynamics of an FBSDE is given by
\begin{align*}
  X_s&=X_{0}+\int_{0}^s\mu(r,X_r,Y_r,Z_r)d r+\int_{0}^s\sigma(r,X_r,Y_r,Z_r)d W_r,\\
  Y_t&=\xi(X_T)-\int_{t}^{T}f(r,X_r,Y_r,Z_r)d r-\int_{t}^{T}Z_r dW_r,
\end{align*}
for $s,t \in [0,T]$ and $X_0 \in \mathbb{R}$, where $(\xi,(\mu,\sigma,f))$ are measurable functions such that 
\begin{align*}
  \xi &\colon\Omega\times\mathbb{R} \to \mathbb{R}, &
  \mu &\colon [0,T]\times\Omega\times\mathbb{R}\times\mathbb{R}\times\mathbb{R}\to \mathbb{R},\\
  \sigma&\colon [0,T]\times\Omega\times\mathbb{R}\times\mathbb{R}\times\mathbb{R}\to \mathbb{R},&
  f&\colon [0,T]\times\Omega\times\mathbb{R}\times\mathbb{R}\times\mathbb{R}\to \mathbb{R},
\end{align*}
Throughout the whole section $\mu$, $\sigma$ and $f$ are assumed to be progressively measurable with respect to $(\mathcal{F}_t)_{t\in[0,T]}$.

A decoupling field comes with an even richer structure than just a classical solution $(X,Y,Z)$.

\begin{definition}\label{def:decoupling field}
  Let $t\in[0,T]$. A function $u\colon [t,T]\times\Omega\times\mathbb{R}\to\mathbb{R}$ with $u(T,\cdot)=\xi$ a.e. is called \emph{decoupling field} for $\fbsde (\xi,(\mu,\sigma,f))$ on $[t,T]$ if for all $t_1,t_2\in[t,T]$ with $t_1\leq t_2$ and any $\mathcal{F}_{t_1}$-measurable $X_{t_1}\colon\Omega\to\mathbb{R}$ there exist progressively measurable processes $(X,Y,Z)$ on $[t_1,t_2]$ such that
  \begin{align}\label{eq:decoupling}
    X_s&=X_{t_1}+\int_{t_1}^s\mu(r,X_r,Y_r,Z_r) d r+\int_{t_1}^s\sigma(r,X_r,Y_r,Z_r)d  W_r,&\nonumber\\
    Y_s&=Y_{t_2}-\int_{s}^{t_2}f(r,X_r,Y_r,Z_r) d r-\int_{s}^{t_2}Z_r d W_r,& \nonumber \\
    Y_s&=u(s,X_s),
  \end{align}
  a.s. for all $s\in[t_1,t_2]$. In particular, we want all integrals to be well-defined.
\end{definition}

Some remarks about this definition are in place.
\begin{itemize}
  \item The first equation in \eqref{eq:decoupling} is called the \emph{forward equation}, the second the \emph{backward equation} and the third will be referred to as the \emph{decoupling condition}.
  \item Note that, if $t_2=T$, we get $Y_T=\xi(X_T)$ a.s.\ as a consequence of the decoupling condition together with $u(T,\cdot)=\xi$. 
  \item If $t_2=T$ we can say that a triplet $(X,Y,Z)$ solves the FBSDE, meaning that it satisfies the forward and the backward equation, together with $Y_T=\xi(X_T)$. This relationship $Y_T=\xi(X_T)$ is referred to as the \emph{terminal condition}. 
\end{itemize}

For the following we need to introduce further notation.

Let $I\subseteq [0,T]$ be an interval and $u: I\times\Omega\times\mathbb{R}\rightarrow \mathbb{R}$ a map such that $u(s,\cdot)$ is measurable for every $s\in I$. We define
\begin{equation*}
  L_{u,x}:=\sup_{s\in I}\inf\{L\geq 0\,|\,\textrm{for a.a. }\omega\in\Omega: |u(s,\omega,x)-u(s,\omega,x')|\leq L|x-x'|\textrm{ for all }x,x'\in\mathbb{R}\},
\end{equation*}
where $\inf \emptyset:=\infty$. We also set $ L_{u,x}:=\infty$ if $u(s,\cdot)$ is not measurable for every $s\in I$. One can show that $L_{u,x}<\infty$ is equivalent to $u$ having a modification which is truly Lipschitz continuous in $x\in\mathbb{R}$.

We denote by $L_{\sigma,z}$ the Lipschitz constant of $\sigma$ w.r.t.\ the dependence on the last component $z$. We set $L_{\sigma,z}=\infty$ if $\sigma$ is not Lipschitz continuous in $z$. 

By $L_{\sigma,z}^{-1}=\frac{1}{L_{\sigma,z}}$ we mean $\frac{1}{L_{\sigma,z}}$ if $L_{\sigma,z}>0$ and $\infty$ otherwise.

For an integrable real valued random variable $F$ the expression $\mathbb{E}_t[F]$ refers to $\mathbb{E}[F|\mathcal{F}_t]$, while $\mathbb{E}_{t,\infty}[F]$ refers to $\esssup\,\mathbb{E}[F|\mathcal{F}_t]$, which might be $\infty$, but is always well defined as the infimum of all constants $c\in[-\infty,\infty]$ such that $\mathbb{E}[F|\mathcal{F}_t]\leq c$ a.s. Additionally, we write $\|F\|_\infty$ for the essential supremum of $|F|$.

In practice it is important to have explicit knowledge about the regularity of $(X,Y,Z)$. For instance, it is important to know in which spaces the processes live, and how they react to changes in the initial value. 

\begin{definition}\label{def:regularity decoupling}
  Let $u\colon [t,T]\times\Omega\times\mathbb{R}\to\mathbb{R}$ be a decoupling field to $\fbsde(\xi,(\mu,\sigma,f))$.
  \begin{enumerate}
   \item We say $u$ to be \emph{weakly regular} if $L_{u,x}<L_{\sigma,z}^{-1}$ and $\sup_{s\in[t,T]}\|u(s,\cdot,0)\|_{\infty}<\infty$.
   \item A weakly regular decoupling field $u$ is called \emph{strongly regular} if for all fixed $t_1,t_2\in[t,T]$, $t_1\leq t_2,$ the processes $(X,Y,Z)$ arising in \eqref{eq:decoupling}  are a.e.\ unique and satisfy
   \begin{equation}\label{strongregul1}
      \sup_{s\in [t_1,t_2]}\mathbb{E}_{t_1,\infty}[|X_s|^2]+\sup_{s\in [t_1,t_2]}\mathbb{E}_{t_1,\infty}[|Y_s|^2]
      +\mathbb{E}_{t_1,\infty}\left[\int_{t_1}^{t_2}|Z_s|^2 d s\right]<\infty,
   \end{equation}
   for each constant initial value $X_{t_1}=x\in\mathbb{R}$. In addition they are required to be measurable as functions of $(x,s,\omega)$ and even weakly differentiable w.r.t. $x\in\mathbb{R}^n$ such that for every $s\in[t_1,t_2]$ the mappings $X_s$ and $Y_s$ are measurable functions of $(x,\omega)$ and even weakly differentiable w.r.t. $x$ such that
   \begin{align}\label{strongregul2}
     &\esssup_{x\in\mathbb{R}}\sup_{s\in [t_1,t_2]}\mathbb{E}_{t_1,\infty}\left[\left|\partial_x X_s\right|^2\right]<\infty, \nonumber\allowdisplaybreaks\\
     &\esssup_{x\in\mathbb{R}}\sup_{s\in [t_1,t_2]}\mathbb{E}_{t_1,\infty}\left[\left|\partial_x Y_s\right|^2\right]<\infty, \nonumber\\
     &\esssup_{x\in\mathbb{R}}\mathbb{E}_{t_1,\infty}\left[\int_{t_1}^{t_2}\left|\partial_x Z_s\right|^2\dx s\right]<\infty.
   \end{align}
   \item We say that a decoupling field on $[t,T]$ is \emph{strongly regular} on a subinterval $[t_1,t_2]\subseteq[t,T]$ if $u$ restricted to $[t_1,t_2]$ is a strongly regular decoupling field for $\fbsde(u(t_2,\cdot),(\mu,\sigma,f))$.
   \end{enumerate}
\end{definition}

Under suitable conditions a rich existence, uniqueness and regularity theory for decoupling fields can be developed. 
\smallskip\\
{\bf Assumption (SLC):} $(\xi,(\mu,\sigma,f))$ satisfies \emph{standard Lipschitz conditions} \textup{(SLC)} if
\begin{enumerate}
  \item $(\mu,\sigma,f)$ are Lipschitz continuous in $(x,y,z)$ with Lipschitz constant $L$,
  \item $\left\|\left(|\mu|+|f|+|\sigma|\right)(\cdot,\cdot,0,0,0)\right\|_{\infty}<\infty$,
  \item $\xi\colon \Omega\times\mathbb{R}\to \mathbb{R}$ is measurable such that $\|\xi(\cdot,0)\|_{\infty}<\infty$ and $L_{\xi,x}<L_{\sigma,z}^{-1}$.
\end{enumerate}

In order to have a notion of global existence we need the following definition:

\begin{definition} We define the maximal interval $I_{\max}\subseteq[0,T]$ of the problem given by  $(\xi,(\mu,\sigma,f))$ as the union of all
intervals $[t,T]\subseteq[0,T]$, such that there exists a weakly regular decoupling field $u$ on $[t,T]$.
\end{definition}

Note that the maximal interval might be open to the left. Also, let us remark that we define a decoupling field on such an interval as a mapping which is a decoupling field
on every compact subinterval containing $T$. Similarly we can define weakly and strongly regular decoupling fields as mappings which restricted to an arbitrary
compact subinterval containing $T$ are weakly (or strongly) regular decoupling fields in the sense of the definitions given above.

Finally, we have global existence and uniqueness on the maximal interval:

\begin{thm}[\cite{Fromm2015}, Theorem 5.1.11,  Lemma 5.1.12 and Corollary 2.5.5]\label{globalexist}
 Let $(\xi,(\mu,\sigma,f))$ satisfy SLC. Then there exists a
unique strongly regular decoupling field $u$ on $I_{\max}$.
Furthermore, either $I_{\max}=[0,T]$ or $I_{\max}=(t_{\min},T]$, where $0 \leq t_{\min} < T$. In the latter case we have
\ben\label{explosion}\lim_{t\downarrow t_{\min}} L_{u(t,\cdot),x}=L_{\sigma,z}^{-1}.\een
Moreover, for any $t\in I_{\max}$ and any initial condition $X_t=x\in\mathbb{R}$ there is a unique solution $(X,Y,Z)$ of the FBSDE on $[t,T]$ satisfying
    \begin{equation*}
      \sup_{s\in[t,T]}\mathbb{E}[|X_s|^2]+\sup_{s\in[t,T]}\mathbb{E}[|Y_s|^2]+\mathbb{E}\left[\int_t^T|Z_s|^2 d s\right]<\infty.
    \end{equation*}
\end{thm}
Equality \eqref{explosion} allows to verify global existence, i.e. $I_{\max}=[0,T]$, via contradiction. We refer to this approach as the method of decoupling fields.

\section{Solving the continuous time problem}\label{solving conti}

In order to apply the results of the previous section we need to allow arbitrary real initial values $X_0=x\in\mathbb{R}$ in \eqref{fbsde}. In addition, we extend the domain of $h$ by setting it to zero whenever $x\leq 0$. Finally, we assume that $h$ is uniformly Lipschitz continuous in $x$. Under this assumption the problem \eqref{fbsde} satisfies (SLC).

We use the method of decoupling fields for proving that there exists a solution
of \eqref{fbsde} on $[0,T]$.
Since the parameters of \eqref{fbsde} satisfy the (SLC), there exists a maximal interval $I_{\max}$ with a weakly regular decoupling field $u$ (see Theorem \ref{globalexist}).

In the following fix $t_0 \in I_{\max}$. Let $(X,Y,Z)=(X^{t_0,x}, Y^{t_0,x},Z^{t_0,x})$ be the solution of \eqref{fbsde} on $[t_0, T]$ with initial value $x \in \R$
such that $Y_t=u(t,X_t)$ a.s. for all $(t,x) \in [t_0, T] \times \R$.

According to strong regularity $u$ is weakly differentiable w.r.t.\ the initial value $x\in\mathbb{R}$. In the following we denote by $u_x$ a version of the weak derivative of $u$ w.r.t.\ $x$ such that it coincides with the classical derivative at all points for which it exists and with $0$ everywhere else. Moreover, the processes $(X,Y,Z)$ are weakly differentiable w.r.t.\ $x$. We can formally differentiate the forward and the backward equation in \eqref{fbsde}. One can verify that one can interchange differentiation and integration and that a chain rule for weak derivatives applies (see Sections A.2 and A.3 in \cite{Fromm2015}). We thus obtain that for every version $(\partial_x X, \partial_x Y,\partial_x Z) = (\partial_x X^{t_0,x}, \partial_x Y^{t_0,x},\partial_x Z^{t_0,x})$ of the weak derivative, such that for every $s\in[t_0,T]$ $(\partial_x X_s, \partial_x Y_s)$ is a weak derivative of $(X_s,Y_s)$, we have for every $t\in[t_0,T]$:
\begin{align}\label{dyn xx}
\partial_x X_t = & 1 - \int_{t_0}^t h_y(s,Y_s)\partial_x Y_s \dx s + \int_{t_0}^t \sigma\cdot \partial_x X_s \dx W_s
\end{align} 
and
\begin{align}\label{dyn yx}
\partial_x Y_t =\partial_x X_T - \int_t^T \partial_x Z_s \dx W_s,
\end{align}
for $\mathbb{P}\otimes\lambda$ - almost all $(\omega,x)\in\Omega\times\IR$.

By redefining $(\partial_x X, \partial_x Y)$ as the right-hand-sides of \eqref{dyn xx} and \eqref{dyn yx} respectively, we obtain processes $(\partial_x X, \partial_x Y)$ that are continuous in time for all $(\omega,x)$ but remain weak derivatives of $X,Y$ w.r.t.\ $x$. From now on, we always assume that $\partial_x X$ and $\partial_x Y$ are continuous in time. We also assume that for fixed $t\in[t_0,T]$ the mappings $\partial_x X_t$ and $\partial_x Y_t$ are weak
derivatives of $X_t$ and $Y_t$ w.r.t. $x\in\mathbb{R}$. In particular $\partial_x X_{t_0}=1$ a.s. for almost all $x\in\IR$.

In order to obtain bounds on the weak derivative $u_x$, we study the process $V_t:= u_x(t,X_t)$, $t \in [t_0, T]$.  

Recall that $Y_t=u(t,X_t)$ a.s. for all $(t,x)\in[t_0,T]\times\mathbb{R}$. Therefore, for fixed $t\in[t_0,T]$, the weak derivatives of the two sides of the equation w.r.t. $x\in\mathbb{R}$
must coincide up to a $\mathbb{P}\otimes\lambda$ - null set. The chain rule for weak derivatives (see Corollary 3.2 in \cite{Ambrosio1990} or Lemma A.3.1.\ in \cite{Fromm2015}) implies, for any fixed $t\in[t_0,T]$, that we have for 
$\mathbb{P}\otimes\lambda$ - almost all $(\omega,x)$  
\begin{align}\label{chain} \partial_x Y_t\mathbf{1}_{\{\partial_x X_t>0\}}=u_x(t,X_t)\partial_x X_t\mathbf{1}_{\{\partial_x X_t>0\}}=V_t\partial_x X_t\mathbf{1}_{\{\partial_x X_t>0\}}. 
\end{align}

Now, choose a fixed $x\in\R$ such that $\partial_x X_{t_0}=1$ a.s., \eqref{chain}, \eqref{dyn xx}, \eqref{dyn yx} are satisfied for almost all $(\omega,t)\in[t_0,T]\times\Omega$ and, in addition, \eqref{chain} is satisfied for $t=t_0$, $\mathbb{P}$ - almost surely.
Note that, since $\partial_x X$, $\partial_x Y$ are continuous in time,
\eqref{dyn xx} and \eqref{dyn yx} in fact hold for all $t\in[t_0,T]$, $\mathbb{P}$ - almost surely.

Observe that $V_t$ is bounded since $u_x$ is bounded. We now turn to the dynamics of $V$.
\begin{lemma} \label{Vdynamic}
The process $(V_t)_{t \in [t_0, T]}$ has a time-continuous version which is an It\^o process. Moreover, there exists a square-integrable progressive process $\wh Z$ such that $(V,\wh Z)$ is the unique solution of the BSDE
\begin{align*}
V_t = 1 - \int_t^T \wh Z_s dW_s - \int_t^T  \left(V^2_s h_y(s,Y_s)-\sigma \wh Z_s\right)\dx s, \qquad t\in[t_0,T].
\end{align*}
\end{lemma}
\begin{proof}
Let $\tau_n = T \wedge \inf\{ t \ge t_0: \partial_x X_t \le \frac{1}{n} \}$. On $[t_0, \tau_n]$ we have $V_t = \partial_x Y_t \frac{1}{\partial_x X_t}$, a.e. Hence $V$ has a version which is an It\^o process on $[t_0, \tau_n]$. We denote the It\^o process decomposition by 
\begin{align*}
V_t = u_x(t_0, x) + \int_{t_0}^t \wh Z_s \dx W_s + \int_{t_0}^t \kappa_s \dx t, \quad t \in [t_0, \tau_n]. 
\end{align*}
The product formula yields, on $[t_0, \tau_n]$, 
\begin{align*}
\dx(V_t \partial_x X_t) = & V_t \left(-h_y(t,Y_t)\partial_x Y_t\dx t+\sigma\cdot \partial_x X_t \dx W_t\right)\\
+ &  \partial_x X_t \left(\kappa_t \dx t+\wh Z_t \dx W_t\right)+\sigma\cdot \partial_x X_t \wh Z_t \dx t.
\end{align*}
Observe that $V_t \partial_x X_t= \partial_x Y_t $. The drift and diffusion coefficients coincide with the coefficients in \eqref{dyn yx}. This implies:
$$ \partial_x Z_t = V_t \sigma \partial_x X_t+\partial_x X_t \wh Z_t $$
and
$$ 0=-V_th_y(t,Y_t)\partial_x Y_t+\partial_x X_t \kappa_t+\sigma\cdot \partial_x X_t \wh Z_t. $$
Using straightforward transformations we obtain
$$ \wh Z_t = \frac{\partial_x Z_t}{\partial_x X_t} - \sigma V_t $$
and 
$$ \kappa_t=V_th_y(t,Y_t)V_t-\sigma \wh Z_t, $$
again on the stochastic interval $[t_0, \tau_n]$. It remains to show that $\tau:=\lim_{n\rightarrow\infty}\tau_n=T$ a.s. To this end note that, according to \eqref{dyn xx}, $\partial_x X_t$ satisfies, on $[t_0,\tau)$, the linear SDE
\begin{align*}
\dx\partial_x X_t = \alpha_t\partial_x X_t \dx t + \sigma\partial_x X_t \dx W_t,
\end{align*} 
where $\alpha_t=-h_y(t,Y_t)V_t$ is uniformly bounded. Consequently, 
$$ \partial_x X_{t\wedge \tau_n}=\exp\left(\int_{t_0}^{t\wedge\tau_n}\left(\alpha_s-\frac{1}{2}\sigma^2\right)\dx s+\int_{t_0}^{t\wedge\tau_n}\sigma\dx W_s\right).$$
Now if $\left(\lim_{n\rightarrow\infty}\partial_x X_{\tau_n}\right)(\omega)=0$ for some $\omega$, 
then $\lim_{n\rightarrow\infty}|W_{t\wedge\tau_n}(\omega)|=\infty$ would hold for the same $\omega$. This, however, is false for almost all $\omega$. In other words, the continuous process $\partial_x X$ does not reach $0$ with probability $1$ and, therefore, 
$\lim_{n\rightarrow\infty}\tau_n=T$ a.s.

In particular $V_t = \partial_x Y_t \frac{1}{\partial_x X_t}$ a.e. and $V$ has a time-continuous version.
\end{proof}

In the following we assume that $V$ refers to the time-continuous version of Lemma \ref{Vdynamic}. Note that there exists a probability measure $Q\sim\mathbb{P}$ such that
$$ \frac{\dx Q}{\dx \mathbb{P}}=\exp\left(\int_{t_0}^T\sigma\dx W_t-\frac{1}{2}\int_{t_0}^T\sigma^2\dx t\right). $$
By Girsanov's theorem $W^Q_t:=W_t-\int_{t_0}^t\sigma\dx s$, $t\in[t_0,T]$, is a Brownian motion w.r.t.\ $Q$. Observe that $V$ satisfies
\begin{align*}
V_t = 1 - \int_t^T \wh Z_s dW^Q_s - \int_t^T  V^2_s h_y(s,Y_s) \dx s, \qquad t\in[t_0,T].
\end{align*}
Working under the new probability measure we now prove:

\begin{lemma}\label{v bdd}
For all $t\in[t_0,T]$ we have a.s. $q\leq V_t \leq 1, $
where $$q:=\exp\left(-T\|h_y\|_\infty\right)\in(0,1).$$
\end{lemma}
\begin{proof}
Define the cut-off function $c(v) := ((v \vee 0) \wedge 1)$ and consider the solution $(\wc V, \wc Z)$ to the Lipschitz BSDE
\begin{align*}
\wc V_t = 1 - \int_t^T \wc Z_s dW^Q_s - \int_t^T  c(\wc V_s)^2 h_y(s,Y_s) \dx s, \qquad t\in[t_0,T].
\end{align*}
It is not a priori clear that $\wc V$ and $V$ are the same. The comparison theorem, applied to $(\wc V, \wc Z)$ and the BSDE with parameters $(1,0)$, implies $\wc V_t \le 1$. A comparison with the BSDE with parameters $\left(1, -\|h_y\|_\infty c\right)$, where $-\|h_y\|_\infty c\leq -c^2 h_y(s,Y_s)$ refers to the generator, yields $\wc V_t \ge q$. Having established that $\wc V$ assumes values in $[0,1]$ we have $c(\wc V_s)=\wc V_s$ such that the bounded processes $\wc V$ and $V$ satisfy the same local-Lipschitz BSDE and are, therefore, the same. As a consequence, we obtain $q\leq V_t \leq 1$.
\end{proof}
Note that we have chosen a version of $V$ such that $V_t=\frac{\partial_x Y_t}{\partial_x X_t}$ for all $t\in[t_0,T]$. Moreover, for $t=t_0$, we have
$V_{t_0}=\frac{\partial_x Y_{t_0}}{1}=u_x(t_0,x)$, a.s.

Since $x$ was chosen arbitrarily outside a $\lambda$ - null set, we have that the $u_x(t_0, \cdot)$ is essentially bounded by  $1$. Since the bound does not depend on $t_0$, by Theorem \ref{globalexist} it must hold that $I_{\max} = [0,T]$, which concludes the proof of well-posedness of the FBSDE \eqref{fbsde}. Moreover, the following holds true:
\begin{propo}\label{mainprop} There exists a unique weakly regular decoupling field $u$ on $[0,T]$ to the problem given by \eqref{fbsde}. In addition to being strongly regular $u$ has the following properties:
\begin{itemize}
\item $u$ is deterministic, i.e. it is a function of $(s,x)$ only.
\item $u(s,x)=x$ whenever $x\leq 0$.
\item The weak derivative $u_x$ takes values in $[q,1]$ only. In particular, $u$ is monotonically increasing in $x$.
\end{itemize}
Finally, for any $X_0=x\in[0,\infty)$ there exists a unique solution $(X,Y,Z)$ of the FBSDE \eqref{fbsde} on $[0,T]$ satisfying
    \begin{equation*}
      \sup_{s\in[0,T]}\mathbb{E}[|X_s|^2]+\sup_{s\in[0,T]}\mathbb{E}[|Y_s|^2]+\mathbb{E}\left[\int_0^T|Z_s|^2 d s\right]<\infty.
    \end{equation*}
This unique solution satisfies $0\leq Y_s=u(s,X_s)\leq X_s$ for all $s\in[0,T]$.
\end{propo} 
\begin{proof}
The existence of $(X,Y,Z)$ follows from Theorem \ref{globalexist}. The fact that  $u_x$ assumes values in $[q,1]$ is a direct consequence of Lemma \ref{v bdd}, which holds for any $t_0\in[0,T]$. Also, $u$ is deterministic because $h$ has this property (see e.g.\ Lemma 2.5.13.\ of \cite{Fromm2015}). 

We next show $u(t_0,x)=x$ for $x\leq 0$ and $t_0\in[0,T]$: In this case the solution to \eqref{fbsde} can be provided explicitly: 
$$ X_t=x\exp\left(\int_{t_0}^t\sigma\dx W_s-\frac{1}{2}\int_{t_0}^t\sigma^2\dx s\right), $$
$$ Y_t=\mathbb{E}\left[X_T|\mathcal{F}_t\right]=X_t, \qquad t\in[t_0,T].$$
It is straightforward to verify that these processes, being non-positive martingales, in fact solve the FBSDE, since $h(t,Y_t)=0$. As a result $u(t_0,x)=u(t_0,X_{t_0})=Y_{t_0}=X_{t_0}=x$.

It remains to show $0\leq Y_s=u(s,X_s)\leq X_s$ for $x\geq 0$: Since $u(s,0)=0$ and since $u$ is increasing in $x$ with the derivative being at most $1$, we have $0\leq u(s,x)\leq x$, for all $x\geq 0$.
\end{proof}

\subsection*{An illustrating example}

In general problem \eqref{fbsde} cannot be solved in closed form and the use of a numerical scheme for coupled FBSDE is necessary to calculate or approximate the decoupling field $u:[0,T]\times\mathbb{R}\rightarrow\mathbb{R}$. Once the decoupling field is obtained, the forward equation, i.e.\ the first equation in \eqref{fbsde}, is straightforward to simulate and $Y$ is obtained from $Y_s=u(s,X_s)$.

However, in special cases the solution can be obtained explicitly: Assume that $h$ is time-homogeneous and linear (for non-negative net equities), i.e.\ $h(s,y)=\gamma y^+$ with some constant $\gamma>0$. Let $t_0\in[0,T]$ and $x=X_{t_0}\in[0,\infty)$. According to Proposition \ref{mainprop} (applied to the interval $[t_0,T]$) the process $Y$ is non-negative, so Lemma \ref{Vdynamic} yields 
$$ V_t = 1 - \int_t^T \wh Z_s dW^Q_s - \int_t^T  V^2_s \gamma\dx s,\qquad t\in[t_0,T]. $$
Since $1$ and $\gamma$ are deterministic it is natural to conjecture that $V$ depends on time only, such that $\wh Z$ vanishes. This would mean that $V$ solves a quadratic ODE the solution for which is straightforward to obtain. In deed, it is straightforward to check that $t\mapsto\frac{1}{1+\gamma (T-t)}$ is a bounded solution to the BSDE satisfied by $V$ in Lemma \ref{Vdynamic}, which means that $\wh Z=0$ and
$$V_t=\frac{1}{1+\gamma (T-t)}, \qquad t\in[t_0,T]. $$
In particular, $u_x(t_0,x)=V_{t_0}=\frac{1}{1+\gamma (T-t_0)}$ for all $x\in[0,\infty)$ and all $t_0\in[0,T]$. This implies 
$$u(t,x)=\frac{x}{1+\gamma (T-t)},\qquad (t,x)\in[0,T]\times[0,\infty)$$
 using Proposition \ref{mainprop}. In particular $Y_0=u(0,X_0)=X_0\left(1+\gamma T\right)^{-1}$.

Next we observe \eqref{fbsde} and conclude that for $X_0\in[0,\infty)$ the corresponding process $X$ on $[0,T]$ is the unique solution to the linear SDE
$$ X_s=X_0-\int_{0}^s \frac{\gamma}{1+\gamma (T-r)} X_r\dx r+\int_{0}^s \sigma X_r\dx W_r. $$
This allows to simulate the gross equity $X$ under the assumption that there is no intrinsic growth or shrinkage. The interested reader is encouraged to apply the It\^o formula to $u$ to verify  that $t\mapsto Y_t=u(t,X_t)$ is then indeed a martingale.\\If an investor wishes to work under the assumption of positive growth, given by an exponential growth rate parameter $\mu>0$, the SDE to simulate would be
$$ X_s=X_0+\int_{0}^s \left(\mu-\frac{\gamma}{1+\gamma (T-r)}\right)X_r\dx r+\int_{0}^s \sigma X_r\dx W_r. $$
Regardless of $\mu\in\IR$ the function $s\mapsto \mathbb{E}[X_s]$ satisfies a linear ODE. Also, the respective value of EbDOs maturing over a time interval $[a,b]\subseteq[0,T]$ is calculated via
$$ \int_a^b\mathbb{E}[h(s,Y_s)]\dx s=\int_a^b\mathbb{E}[h(s,u(s,X_s))]\dx s=
\int_a^b\frac{\gamma\mathbb{E}[X_s]}{1+\gamma (T-s)}\dx s. $$

Note that in the above example the volatility parameter $\sigma$ influences the volatility of the processes $X,Y$ but not their expected values. Neither is the expected payoff (resp.\ value of EbDOs) influenced. This is characteristic for the linear case only.
\bibliography{quellen}{}
\bibliographystyle{abbrv}

\end{document}